\documentclass[10pt,twocolumn,journal,twoside]{IEEEtran}
\usepackage{hyperref}

\newtheorem{theorem}{Theorem}
\newtheorem{lemma}{Lemma}
\newtheorem{corollary}{Corollary}

\newtheorem{problem}{Problem}
\newtheorem{remark}{Remark}

\begin{document}
\title{Propelinear 1-perfect codes from quadratic functions}
\author{Denis~S.~Krotov 
       and~Vladimir~N.~Potapov
\thanks{D. Krotov and V. Potapov are with the Sobolev Institute of Mathematics,
pr. Akademika Koptyuga 4, Novosibirsk 630090, Russia, and 
with the Novosibirsk State University, Pirogova 2, Novosibirsk 630090, Russia;
e-mail: krotov@math.nsc.ru, vpotapov@math.nsc.ru.}
\thanks{Manuscript received October 19, 2013; revised January 22, 2014.
The work 
was supported in part 
by the RFBR (project 13-01-00463), 
by the Ministry of education 
and science of Russian Federation 
(project 8227),
and by the Target program of SB RAS 
for 2012-2014 (integration project No. 14). 
The material in this paper was presented in part at the Mal'tsev Meeting, Novosibirsk, Russia, November 12-16, 2012.}%
\thanks{Copyright \copyright 2014 IEEE. Personal use of this material is permitted.  However, permission to use this material for any other purposes must be obtained from the IEEE by sending a request to pubs-permissions@ieee.org}%
\thanks{Digital Object Identifier 
\href{http://dx.doi.org/10.1109/TIT.2014.2303158}{10.1109/TIT.2014.2303158}}}

\markboth{IEEE~Transactions~on~Information~Theory,\quad
\url{http://dx.doi.org/10.1109/TIT.2014.2303158}}%
{Krotov and 
Potapov: Propelinear 1-perfect codes from quadratic functions}

\maketitle

\begin{abstract}
Perfect codes obtained 
by the Vasil'ev--Sch\"onheim construction 
from a linear base code 
and quadratic switching functions
are transitive and, moreover, propelinear. 
This gives at least $\exp(cN^2)$ 
propelinear $1$-perfect codes of length $N$
over an arbitrary finite field, 
while an upper bound
on the number of transitive codes is 
$\exp(C(N\ln N)^2)$.
\end{abstract}

\section{Introduction}
\IEEEPARstart{U}{sually}, a group code is defined as a subgroup
of the additive group 
of a finite vector space.
There are alternative approaches 
\cite{Nechaev82,RBH:89,%
Nechaev:Kerdock_cyclic_form,%
HammonsOth:Z4_linearity,Carlet:Z2k-linear,%
HonNech,Kro:2007Z2k} 
that allow
to relate the codewords of a code 
with the elements 
of some group.
Usually, the mapping 
from the group to the code
is required to satisfy 
some metric properties,
because the distance is 
what is very important 
for error-correcting codes.
One of the approaches considers 
so-called propelinear codes, 
introduced in \cite{RBH:89} 
for the binary space.
The codewords of a propelinear code $C$
are in one-to-one correspondence 
with a group $G$ of isometries of the space
that acts regularly on the code itself.
In other words, 
given some fixed codeword $v\in C$
(say, the all-zero word), 
every other codeword
can be uniquely written as $g(v)$, $g\in G$.
Every propelinear code is transitive; 
that is, it is an orbit
of a group of isometries of the space
(for a transitive code in general, 
this group is not required to act regularly).

In the current paper, 
we will prove that the number 
of nonequivalent
propelinear codes with the same parameters,
namely, 
the parameters of $1$-perfect codes
over an arbitrary finite field,
grows at least exponentially with respect 
to the square of the code length 
(Corollary~\ref{cor:lb}). 
By the order of the logarithm, 
this number is comparable 
with the total number 
of propelinear codes (Theorem~\ref{p:up}).
In contrast, there is only one (up to equivalence) 
linear $1$-perfect code
for each admissible length, but the number of non-linear
$1$-perfect codes grows doubly-exponentially 
\cite{Vas:nongroup_perfect,Schonheim68}.

For the case $q=2$, 
an exponential lower bound
(with respect to the square root 
of the code length, to be more accurate)
on the number of transitive and the number of propelinear $1$-perfect codes 
was firstly established in
\cite{Pot:trans} and \cite{BRMS:num_prop}, 
respectively.
Here, we will show how to improve the lower bound 
and generalize it to an arbitrary prime power $q$,
using a rather simple construction.
Some other constructions 
of transitive and propelinear perfect codes 
can be found in 
\cite{BorRif:1999},
\cite{Kro:2000:Z4_Perf},
\cite{Sol:2005transitive},
\cite{BRMS:2012}.

Section~\ref{s:pre} contains definitions and auxiliary lemmas.
In Section~\ref{s:main}, we formulate the main results of the paper.
The main theorem is proven in Section~\ref{s:proof}.
In Section~\ref{s:rem}, 
we consider some remarks and examples
concerning the structure of the group 
related to a propelinear code,
survey the transitive (propelinear) Vasil'ev codes of length $15$,
and discuss a problem about 
functions that can result in transitive codes.
\section{Preliminaries}\label{s:pre}
Let $F$ be a finite field of order $q$, where $q$ is a power of prime; 
let $F^n$ be the vector space 
of all words of length $n$ over the alphabet $F$.
An arbitrary subset of $F^n$ 
is referred to as a code. 
A code is linear if it is a vector subspace of $F^n$.
A code $C\subset F^n$ is called $1$-perfect 
if for every word $v$ from $F^n$ 
there is exactly one $c$ in $C$
agreeing with $v$ in at least $n-1$ positions.
It is well known that $1$-perfect codes exist if and only if
$n=(q^k-1)/(q-1)$ for some integer $k$, see e.g. \cite{MWS}. 
%
\subsection{Vasil'ev--Sch\"onheim construction}

Let $H\subset F^n$, and let $f:H \to F$ 
be an arbitrary function. Define the set
\begin{IEEEeqnarray*}{rcrl}
 \displaystyle
C(H,f) &{}={}& 
\Big\{ ((v_\alpha)_{\alpha\in F},z) \, : \, 
   v_\alpha\in F^n, 
   \ \sum_{\alpha\in F}v_\alpha = c \in H,  \\ 
  && z = \sum_{\alpha\in F} \alpha|v_\alpha| + f(c) 
 &\Big\}
\end{IEEEeqnarray*}
where $(v_\alpha)_{\alpha\in F}$ is treated 
as the concatenation of the words $v_\alpha$ 
(which will be referred to as \emph{blocks})
in some prefixed order,
$|v_\alpha|$ is the sum of all $n$ elements 
of $v_\alpha$.
If $H$ is a $1$-perfect code, 
then $C(H,f)$ is a $1$-perfect code in  $F^{qn+1}$,
known as a Sch\"onheim code \cite{Schonheim68} 
(in the case $q=2$, as a Vasil'ev code 
\cite{Vas:nongroup_perfect}). 
Clearly, the set $C(H,f)$ essentially depends 
on the choice of the function $f$:

\begin{lemma}\label{l:dif}
For a fixed $H$, 
different functions $f$ result in different $C(H,f)$.
\end{lemma}
\begin{IEEEproof}
The graph of the function $f$ can be reconstructed 
from the set $C(H,f)$:
\begin{IEEEeqnarray*}{rcrl}
\{ (x,f(x)) \,:\, x\in H\} &{}={}&
\Big\{ \Big(\sum_{\alpha\in F}v_\alpha,~
  z-\sum_{\alpha\in F} \alpha|v_\alpha|\Big) \,:\, \\
  &&  ((v_\alpha)_{\alpha\in F},z)\in C(H,f)&\Big\}.
\end{IEEEeqnarray*}
Hence, $C(H,f)=C(H,f')$ implies $f=f'$.
\end{IEEEproof}

\subsection{Automorphisms, equivalence, 
transitivity, propelinearity}

The \emph{Hamming graph} $G(F^n)$ is defined
on the vertex set $F^n$; 
two words are connected by an edge 
if and only if they differ 
in exactly one position.
It is known 
(see, e.g., \cite[Theorem~9.2.1]{Brouwer}) 
that every automorphism $\Pi$
of $G(F^n)$ is composed from 
a coordinate permutation $\pi$
and alphabet permutations $\psi_i$
in each coordinate:
$\Pi(x)=(\psi_1(x_{\pi^{-1}(1)}),
  \ldots,\psi_n(x_{\pi^{-1}(n)}))$.
Two codes are said to be \emph{equivalent}
if there is an automorphism of $G(F^n)$ 
that maps one of the codes into the other.
Note that 
the algebraic properties
of the code, such as being a 
linear or affine subspace,
are not invariant with respect to 
this combinatorial equivalence, in general.
The \textit{automorphism group} 
$\mathrm{Aut}(C)$ of a code $C$
consists of all automorphisms of $G(F^n)$
that stabilize (fix set-wise) $C$.
A code $C$ containing the all-zero 
word $\bar 0$ is \emph{transitive} 
if for every codeword $a$ 
there exists $\varphi_a \in \mathrm{Aut}(C)$ 
that sends $\bar 0$ to $a$. 
If, additionally, the set 
$\{\varphi_a\,:\,a\in C\}$ 
is closed under composition 
(that is, for all $a$ and $b$ from $C$ 
we have $\varphi_a\varphi_b = \varphi_c$, 
where $c=\varphi_a\varphi_b(\bar 0)$),
then $C$ is a \emph{propelinear} code, 
see e.g. \cite{BRMS:num_prop}.

\subsection{Quadratic functions}

Assume $H$ is a subspace of $F^n$.
A function $f:H \to F$ is called \textit{quadratic} 
if it can be represented 
as a polynomial of degree at most $2$.

We will use the following elementary 
property of the quadratic functions 
(actually, it is a characterizing property).
\begin{lemma}\label{l:qua}
Let $H$ be a subspace of $F^n$.
If $f:H \to F$ is a quadratic function,
then for every $c\in H$ there exist 
$\beta_0^c,\beta_1^c,\ldots,\beta_n^c \in F$ such that
\begin{IEEEeqnarray}{c}\label{eq:qua}
f(x+c) = 
f(x) + \beta_0^c + \beta_1^c x_1 + 
            \ldots + \beta_n^c x_n \\ \nonumber
 \mbox{for all $x=(x_1,\ldots,x_n)\in H$}.
\end{IEEEeqnarray}
Moreover, $\beta_i^c$, $i\in\{1,\ldots,n\}$, 
depends linearly on $c$: 
$$\beta_i^{c+d}=\beta_i^c+\beta_i^d.$$
\end{lemma}
\begin{IEEEproof}
The difference of $(x_i+c_i)(x_j+c_j)$ 
and $x_ix_j$ has degree at most $1$. 
Moreover, the coefficients at $x_i$ and $x_j$ 
in this difference depend linearly on $c$.
Hence, the same is true for the difference of $P(x+c)$ 
and $P(x)$ for every polynomial $P$
of degree at most $2$.
\end{IEEEproof}

\begin{lemma}\label{l:nqua}
Let $H$ be an $m$-dimensional subspace of $F^n$.
There are at least $q^{{m^2}/2}$ 
different quadratic functions from $H$ to $F$. 
\end{lemma}
\begin{IEEEproof}
Obviously, a linear transformation 
of the space does not affect 
to the property of a function
to be quadratic. 
Hence, we can assume 
without loss of generality that 
$H$ consists of 
the words of length $n$ with zeroes in the last $n-m$ positions. 
Then, the number of different quadratic functions 
is the number of polynomials of degree at most $2$ 
in the first $m$ variables, 
i.e., $q^{m(m-1)/2+m+m+1}$ for $q>2$ 
and $q^{m(m-1)/2+m+1}$ for $q=2$ 
(when $x_i^2\equiv x_i$ (mod $2$)).
\end{IEEEproof}

\section{Main results}\label{s:main}
\subsection{Lower bound}
In the Section~\ref{s:proof}, 
we will prove the following theorem.
\begin{theorem}\label{th} 
If $H\subset F^n$ is a linear $q$-ary code 
and $f:H \to F$ is a quadratic function, $f(\bar 0)=0$,
then $C(H,f)$ is a propelinear code of length $N=qn+1$.
\end{theorem}

\begin{corollary}\label{cor:lb}
The number of nonequivalent propelinear 
$1$-perfect $q$-ary codes 
of length $N=(q^k-1)/(q-1)$ obtained by
the Vasil'ev--Sch\"onheim construction 
is at least $q^{\frac{N^2}{2q^2}+O(N\ln N)}$.
\end{corollary}
\begin{IEEEproof}
As follows from Theorem~\ref{th}, 
Lemma~\ref{l:dif}, and Lemma~\ref{l:nqua},
the number of different propelinear $1$-perfect 
codes of type $C(H,f)$ is at least 
$q^{\frac{m^2}2}$, where $m=n-\log_q (nq-n+1)$
and $n$ is the length of $H$. 
Since $N=qn+1$, we see that 
$q^{\frac{m^2}2}=q^{\frac{N^2}{2q^2}+O(N\ln N)}$.
To evaluate the number of nonequivalent codes,
we divide this number by the number 
$N!(q!)^N=q^{O(N\ln N)}$ 
of all automorphisms of $F^N$ and find that this
does not affect on the essential part of the formula.
\end{IEEEproof}

\subsection{Upper bound}

To evaluate how far our lower bound 
on the number of transitive (propelinear)
$1$-perfect codes can be 
from the real value, we derive an upper bound:
\begin{theorem}\label{p:up}
{\rm (a)} The number of different 
transitive codes in $F^N$ 
does not exceed $2^{(N \log_2 N)^2 (1+o(1))}$.
{\rm (b)} The number of different 
propelinear codes in $F^N$
does not exceed $q^{N^2 \log_2 N (1+o(1))}$. 
\end{theorem}
\begin{IEEEproof}
Since every subgroup of $\mathrm{Aut}(F^N)$ 
is generated by at most $\log_2 |\mathrm{Aut}(F^N)|$
elements, the number of subgroups is less than
$|\mathrm{Aut}(F^N)|^{\log_2|\mathrm{Aut}(F^N)|}
= 2^{(N \log_2 N)^2 (1+o(1))}$ 
(recall that 
$|\mathrm{Aut}(F^N)|=(q!)^N N! = N^{N(1+o(1))}$).
Since every transitive code $C$ 
containing $\bar 0$ is uniquely determined
by its automorphism group (indeed, 
$C$ is the orbit of $\bar 0$ under $\mathrm{Aut}(F^N)$),
statement (a) follows. 

The automorphisms assigned to the codewords 
of a propelinear code $C$ form a group 
of order $|C|\leq q^N$. 
It is generated by 
at most $\log_2 q^N=N\log_2 q$ elements; 
each of them can be chosen in less than 
$|\mathrm{Aut}(F^N)|=N^{N(1+o(1))}$ ways;
(b) follows.
\end{IEEEproof}
\section{Proof of Theorem~\ref{th}}\label{s:proof}
Let $H\subset F^n$ be a linear code 
and let $f:H \to F$ be a quadratic function.
The key point in the proof 
is the following simple statement.

\begin{lemma}\label{l:1} 
Let $f'(x)=f(x)+\beta x_j$ for some  
$j\in \{1,\ldots,n\}$, $\beta\in F$.
Then $ C(H,f') = \Pi_j^{\beta} C(H,f) $ 
where $\Pi_j^{\beta}$ 
is the coordinate permutation that sends 
the $j$'th coordinate of the block 
$v_{\alpha+\beta}$
to the $j$'th coordinate of the block $v_\alpha$
for all $\alpha\in F$ 
and fixes the other coordinates.
\end{lemma}

\begin{IEEEproof} 
Let us consider the codeword 
$x=((v_\alpha)_{\alpha\in F},z)$ of $C(H,f)$. 
It satisfies 
$z=\sum_{\alpha\in F} \alpha|v_\alpha| + f(c)$.
After the coordinate permutation $\Pi_j^{\beta}$, 
we obtain the word 
$y=\Pi_j^{\beta} x=((u_\alpha)_{\alpha\in F},z)$ 
where for all $\alpha$ 
the word $u_\alpha$ coincides with $v_\alpha$ 
in all positions except the $j$th, 
$u_{\alpha,j}$ which is equal to 
$v_{\alpha+\beta,j}$. Now we have
\begin{IEEEeqnarray*}{r;l}
  z =& \sum_{\alpha\in F} \alpha|v_\alpha| + f(c) 
\\=&
\sum_{\alpha\in F}  \sum_{k\neq j} \alpha v_{\alpha,k} +
\sum_{\alpha\in F}  \alpha v_{\alpha,j} + f(c) 
\\=&
\sum_{\alpha\in F}  \sum_{k\neq j} \alpha u_{\alpha,k} +
\sum_{\alpha\in F}  \alpha u_{\alpha-\beta,j} + f(c)
\\=&  
\sum_{\alpha\in F}  \sum_{k\neq j} \alpha u_{\alpha,k} +
\sum_{\alpha\in F}  (\alpha+\beta) u_{\alpha,j} + f(c)
\\=& 
\sum_{\alpha\in F}  \sum_{k=1}^n \alpha u_{\alpha,k} +
\beta\sum_{\alpha\in F}   u_{\alpha,j} + f(c) 
\\=& 
\sum_{\alpha\in F} \alpha|u_\alpha| + f(c) +  \beta c_j, 
\end{IEEEeqnarray*}
(we used that 
$c=(c_1,\ldots,c_n)=\sum v_\alpha = \sum u_\alpha$)
which proves that $\Pi_j^{\beta} (x) \in C(H,f')$.
\end{IEEEproof}

Now denoting 
$\Pi^c = 
\Pi_1^{\beta_1^c}\Pi_2^{\beta_2^c}
  \ldots\Pi_n^{\beta_n^c}$,
where the coefficients 
$\beta_{j}^{c}$ are from (\ref{eq:qua}),
we get the following fact, 
which immediately proves 
the transitivity of the code:

\begin{lemma}\label{l:2} 
For every codeword 
$w=((w_\alpha)_{\alpha\in F},z)$ of $C(H,f)$, 
the transform
$\Phi_w(v) = w+\Pi^c(v)$, 
where $c=\sum_{\alpha\in F}w_\alpha$, 
is an automorphism of $C(H,f)$, which sends
the all-zero word to $w$.
\end{lemma}

\begin{IEEEproof}
Consider $v = ((v_\alpha)_{\alpha\in F},s)$ 
from $C(H,f)$.
It satisfies 
$s=\sum_{\alpha\in F} \alpha|v_\alpha| + f(d)$, 
where $d=\sum_\alpha v_\alpha$.
Applying Lemma~\ref{l:1} with $j=1,\ldots,n$, we see that 
$\Pi^c (v) = ((u_\alpha)_{\alpha\in F},s)$ satisfies 
$s=\sum_{\alpha\in F} \alpha|u_\alpha| + f(d) +
\beta_1^c d_1 + \ldots + \beta_n^c d_n$, 
where $d=(d_1,\ldots,d_n)=\sum_\alpha u_\alpha$.
Adding $w=((w_\alpha)_{\alpha\in F},z)$, we obtain
$w+\Pi^c (v) = ((w_\alpha+u_\alpha),r)$, where 
\begin{IEEEeqnarray*}{l;c;l}
r &=& \sum_{\alpha\in F} \alpha|u_\alpha| + f(d) 
+ \beta_1^c d_1 + \ldots + \beta_n^c d_n \\ 
 &&{} + \sum_{\alpha\in F} \alpha|w_\alpha| + f(c)
 \\  &=& 
\sum_{\alpha\in F} \alpha|u_\alpha+w_\alpha| 
 + f(d+c) - \beta_0^c + f(c).
\end{IEEEeqnarray*}
But $f(c) = f(\bar 0)+\beta_0^c$, 
as we see from (\ref{eq:qua}).
Since $f(\bar 0)=0$, we have proved that 
$w+\Pi^c (v)$ belongs to $C(H,f)$.
\end{IEEEproof}

So, we get the transitivity. It remains to prove that
the set of $\Phi_w$, $w\in C(H,f)$ is closed under composition.

\begin{lemma}\label{l:pipi}
 For every $c,d\in H$ the composition $\Pi^c\Pi^d$ equals~$\Pi^{c+d}$.
\end{lemma}
\begin{IEEEproof}
 As follows directly from the definitions of $\Pi^{c}$ and~$\Pi^\beta_i$,
\begin{IEEEeqnarray*}{r;c;l} \Pi^c\Pi^d
 &=& \Pi_1^{\beta_1^c}\ldots\Pi_n^{\beta_n^c}\Pi_1^{\beta_1^d}\ldots\Pi_n^{\beta_n^d}\\
 &=& \Pi_1^{\beta_1^c}\Pi_1^{\beta_1^d}\Pi_2^{\beta_2^c}\Pi_2^{\beta_2^d}\ldots\Pi_n^{\beta_n^c}\Pi_n^{\beta_n^d}.
\end{IEEEeqnarray*} 
By the definition of $\Pi^\beta_i$, we have
 $\Pi_i^{\beta_i^c}\Pi_i^{\beta_i^d}=\Pi_i^{\beta_i^c+\beta_i^d}.$
 But, by Lemma~\ref{l:qua}, $\beta_i^c+\beta_i^d=\beta_i^{c+d}$. Finally, we have
 $\Pi^c\Pi^d = \Pi_1^{\beta_1^{c+d}}\ldots\Pi_n^{\beta_n^{c+d}} = \Pi^{c+d}.$
 \end{IEEEproof}

Now, consider $w=((w_\alpha)_{\alpha\in F},z)$ and
$v = ((v_\alpha)_{\alpha\in F},s)$ form $C(H,f)$. 
Denote $c=\sum_{\alpha}w_\alpha$ and $d=\sum_{\alpha}v_\alpha$; 
observe that the permutation $\Pi^c$ 
will not change the value of the last sum.
Then, 
\begin{IEEEeqnarray*}{r;c;l} 
 \Phi_w\Phi_v(\cdot)
&=&  w+\Pi^c(v+\Pi^d(\cdot))\\
&=& w+\Pi^c(v)+\Pi^c(\Pi^d(\cdot))
= u+\Pi^e(\cdot),
\end{IEEEeqnarray*}
where 
$u=((u_\alpha)_{\alpha\in F},t)=w+\Pi^c(v)$, 
$e=\sum_{\alpha}u_\alpha=c+d$.
This completes the proof of the theorem.

\section{Remarks, examples, and further research}\label{s:rem}
\subsection{On the group related to $C(H,f)$}
As follows from the definition, 
to every codeword $v$ of a propelinear code $C$
there corresponds an automorphism $\Phi_v$ of $C$
and the set $\{\Phi_v : v\in C\}$ forms
a subgroup of the automorphism group of $C$.
Although such a subgroup, a \emph{propelinear structure}, 
is not unique in general 
(see \cite{BRMS:2012} and also Remark~\ref{rem:nonunique} below),
in the previous section we explicitly defined 
a variant of the choice of $\Phi_v$ for every $v\in C(H,f)$.
Below, we provide two remarks with examples 
about the propelinear structure
defined in the previous section.

\begin{remark}\label{rem:order}
For every $v\in C(H,f)$, the element $\Phi_v$ 
has order $1$, $p$, or $p^2$, 
where $p$ is the prime that divides $q$. 
Indeed, every permutation $\Pi^c$ is of order
$1$ or $p$; hence, $(\Phi_v)^p$ corresponds to 
the identity permutation and has order $1$ or $p$.
\end{remark}

As an example, we consider the (non-perfect)
code $C(H,f)$ constructed with the following parameters: 
$q=2$, $n=2$, $H=F^2$, $f(x_1,x_2)=x_1x_2$.
From (\ref{eq:qua}) we find 
$\beta^{01}_1 = 1$, $\beta^{01}_2 = 0$, 
$\beta^{10}_1 = 0$, $\beta^{10}_2 = 1$, 
$\beta^{11}_1 = 1$, $\beta^{11}_2 = 1$.
The group of automorphisms related with the propelinear code $C(H,f)$
is generated by three elements $\Phi_{u}$, $\Phi_{v}$, $\Phi_{w}$ with
$u=(11\,00\,1)$, 
$v=(10\,00\,0)$, 
$w=(10\,10\,1)$ and the corresponding coordinate permutations
$\Pi^{11}=(13)(24)$, $\Pi^{10} = (24)$, $\Pi^{00} = \mathrm{Id}$.
The first element $\Phi_{u}$ generates a cycle 
with the corresponding codewords 
$(00\,00\,0)$, $(11\,00\,1)$, 
$(11\,11\,0)$, $(00\,11\,1)$.
The second generating element $\Phi_{v}$ adds four more codewords:
$(10\,00\,0)$, $(00\,01\,1)$, 
$(01\,11\,0)$, $(11\,10\,1)$; 
the corresponding automorphisms are of order $2$.
The group generated by $\Phi_{u}$ and $\Phi_{v}$ 
is described by the orders 
of $\Phi_{u}$, $\Phi_{v}$ and the identity 
$\Phi_{v}\Phi_{u}\Phi_{v} = (\Phi_{u})^{-1}$,
and it is isomorphic to the dihedral group $D_4$.
The last generating element $\Phi_{w}$ 
commutes with all other elements and has order $2$.
It follows that the group of automorphisms related with
$C(H,f)$ is isomorphic to the direct product
$D_4\times Z_2$,
where $Z_2$ is the cyclic group of order $2$.

\begin{remark}\label{rem:nonunique}
If $H\ne F^n$, then there is more than one quadratic 
representations of every quadratic function on $H$.
The coefficients $\beta^c_i$ and, as follows,
the subgroup $\{\Phi_v : v\in C\}$ of the automorphism 
group of the code
depend on the representation;
so, there are several propelinear structures 
corresponding to the same code $C(H,f)$.
For example, the all-zero function over $H=\{000,111\}$
($q=2$) can be represented as $f(x_1,x_2,x_3)=0$
or, e.g., as $f(x_1,x_2,x_3)=x_1x_2+x_1x_3$.
The resulting code is the same 
(a $1$-perfect Hamming code of length $7$);
but in the first case, the group is
isomorphic to $Z_2^4$,
while the second representation 
leads to a group isomorphic to $Z_4\times Z_2^2$.
The general fact that several propelinear 
structures can correspond to the same
(perfect) code was well demonstrated in \cite{BRMS:2012}.
\end{remark}

\subsection{Transitive Vasil'ev codes of length $15$}\label{ss:15}
There are $201$ nonequivalent transitive $1$-perfect codes of length $15$
\cite[Table III]{OPP:15}.
Five of these codes are Vasil'ev codes, including the linear one;
their description can be found in \cite{Malyu:2004} 
(the four nonlinear codes
are denoted by $V4$, $V4^0$, $V4^1$, and $V22^02^1$).
Let $H$ be spanned by the words 
$u_1 = 1010101$, $u_2=0111100$, $u_4=0001111$, $u_{0}=1111111$.
Define the functions $f^{V4}$, $f^{V4^0}$, $f^{V4^1}$, $f^{V22^02^1}:H\to\{0,1\}$
by their sets of zeros 
$\{\bar 0, u_0,u_1,u_0+u_1\}$,
$\{\bar 0, u_1,u_2,u_1+u_2\}$,
$\{\bar 0, u_0+u_1,u_0+u_2,u_1+u_2\}$,
$\{\bar 0, u_0,u_1,u_2,u_4,u_0+u_1+u_2+u_4\}$, respectively.
Then the codes $C(H,f^{V4})$, $C(H,f^{V4^0})$, $C(H,f^{V4^1})$, $C(H,f^{V22^02^1})$
are representatives of the four equivalence classes of 
nonlinear transitive Vasil'ev codes of length $15$.
All these codes are propelinear \cite{BRMS:2012}.
Moreover, it can be directly checked that the functions 
are quadratic:
\begin{IEEEeqnarray*}{r;c;l}
f^{V4}(x) &=&x_2x_4+x_2x_6+x_4x_6+x_2+x_4+x_6,\\
f^{V4^0}(x) &=&x_1x_6+x_2x_6+x_3x_6+x_1+x_2+x_3+x_6,\\
f^{V4^1}(x)&=&x_3 x_5+x_3+x_5,\\
f^{V22^02^1}(x)&=&
x_3x_4+x_3x_5+x_3x_7+x_4x_5+x_4x_7+x_5x_7
\\&&{}+x_3+x_4+x_5+x_7;
\end{IEEEeqnarray*}
so, the corresponding codes meet the hypothesis of Theorem~\ref{th}.
Therefore, all transitive Vasil'ev codes of length $15$ belong to the class
considered in the current paper.

\subsection{Transitive functions}\label{ss:fin}
For further development of the topic, 
it would be interesting to consider 
a wider class of functions resulting 
in transitive (propelinear) codes.
Such functions should have properties 
similar to transitivity (propelinearity) of codes:

\begin{problem} 
For a vector space $V$ 
and a group $\mathcal A$ 
of linear permutations of $V$,
find non-quadratic functions $f$ 
such that for every $c$ from $V$ 
there exists $\mu\in \mathcal A$
meeting $f(\mu(x)+c)=f(x)+l(x)$ 
for some affine $l$.
For instance, for constructing transitive (propelinear)
$1$-perfect codes as above, 
we can take $V=H$ and 
$\mathcal A\subset \mathrm{Aut}(H)$.
\end{problem} 
\newpage
\providecommand\href[2]{#2} \providecommand\url[1]{\href{#1}{#1}}
  \def\DOI#1{{\small {DOI}:
  \href{http://dx.doi.org/#1}{#1}}}\def\DOIURL#1#2{{\small{DOI}:
  \href{http://dx.doi.org/#2}{#1}}}

\end{document}